\documentclass[runningheads,a4paper]{llncs}

\usepackage{amssymb}
\usepackage{graphicx}

\usepackage{url}
\urldef{\mailsa}\path|{alfred.hofmann, ursula.barth, ingrid.haas, frank.holzwarth,|
\urldef{\mailsb}\path|anna.kramer, leonie.kunz, christine.reiss, nicole.sator,|
\urldef{\mailsc}\path|erika.siebert-cole, peter.strasser, lncs}@springer.com|

\usepackage{amsmath}
\usepackage{epstopdf}
\usepackage{caption}
\usepackage{subcaption}
\usepackage{framed}
 \usepackage[usenames,dvipsnames]{pstricks}
 \usepackage{pst-grad} % For gradients
 \usepackage{pst-plot} % For axes

\long\def\comment#1{}

\begin{document}

\mainmatter  % start of an individual contribution

\title{Incentive Design for Ridesharing with Uncertainty}

\titlerunning{Incentive Design for Ridesharing with Uncertainty}

%\author{Alfred Hofmann%
%\thanks{Please note that the LNCS Editorial assumes that all authors have used
%the western naming convention, with given names preceding surnames. This determines
%the structure of the names in the running heads and the author index.}%
%\and Ursula Barth\and Ingrid Haas\and Frank Holzwarth\and\\
%Anna Kramer\and Leonie Kunz\and Christine Rei\ss\and\\
%Nicole Sator\and Erika Siebert-Cole\and Peter Stra\ss er}

\author{Dengji Zhao \and Sarvapali D. Ramchurn \and Nicholas R. Jennings
}
\authorrunning{Zhao \textit{et al}.}


\institute{
Electronics and Computer Science\\
University of Southampton\\
Southampton, SO17 1BJ, UK\\
\{d.zhao, sdr, nrj\}@ecs.soton.ac.uk
}


\toctitle{Lecture Notes in Computer Science}
%\tocauthor{Authors' Instructions}
\maketitle

\begin{abstract}
We consider a ridesharing problem where there is uncertainty about the completion of trips from both drivers and riders. Specifically, we study ridesharing mechanisms that aim to incentivize commuters to reveal their valuation for trips and their probability of undertaking their trips.
Due to the interdependence created by the uncertainty on commuters' valuations, we show that the Groves mechanisms are not ex-post truthful even if there is only one commuter whose valuation depends on the other commuters' uncertainty of undertaking their trips. %, e.g., a rider's valuation for a ride depends on the driver's probability of undertaking the trip.
To circumvent this impossibility, we propose an ex-post truthful mechanism, the best incentive we can design without sacrificing social welfare in this setting. Our mechanism pays a commuter if she undertakes her trip, otherwise she is penalized for not undertaking her trip. Furthermore, we identify a sufficient and necessary condition under which our mechanism is ex-post truthful.
\end{abstract}

\section{Introduction}
\noindent Ridesharing has been touted as a key mechanism to optimise transportation systems since the 1940s. By having multiple road users share a car, it may significantly reduce fuel costs, traffic congestion, and CO$_2$ emissions~\cite{CHAN_2012}. %\cite{Jacobson_2009,CHAN_2012}. 
Moreover, a number of private ridesharing services such as \textit{Uber} and \textit{Lyft} have introduced real-time online booking systems to allow consumers to book rides seamlessly. % which recently attracted more attentions from both business investigators and the end users largely because of the availability of technologies, say, smart phones and social networking services, and more importantly our great desire for a better living environment. More and more people commute by other transport modes such as ridesharing rather than driving. Instead of owning a car, many people are inclined to share or rent cars especially those that are environment-friendly such as electric vehicles~\cite{GEVO_2013}, which makes ridesharing become more natural. 
Despite such efforts, however, the number of users of ridesharing services has not significantly grown over the years.\footnote{The share of US workers commuting by ridesharing/carpooling has declined from 20.4\% in 1970 to just 9.7\% in 2011 (the US Census).} There are a number of reasons for this but here we focus on one of the key challenges: these actors must find it more convenient to share a ride rather than take their own car or other transports. An important factor that affects convenience is the ability to plan trips at short notice, but also to be able to deal with ride cancellations. Unfortunately, in current ridesharing services, if there is a no-show of a driver or a rider, the rider or the driver may be significantly penalized (e.g., \textit{Uber} and \textit{Lyft} charge a user $5$ to $10$ dollars for cancelling a ride in the US). Moreover, both \textit{Uber} and \textit{Lyft} operate like taxi companies with dedicated drivers and standard but low fare rates. They indeed motivate riders to use their services, but hardly involve many low-occupancy vehicles on the roads. Therefore, it is crucial that ridesharing systems are designed to incentivize both riders and drivers to use the services while accounting for the execution uncertainty of their trips. 

To date, researchers have proposed auction-based ridesharing systems that allow more people to participate and also shift the effort of arranging rides from the users to the system~\cite{Kamar_2009,Kleiner_2011,Zhao_2014}. 
{Given people's travel plans/preferences, these auction-based systems automatically compute their sharing schedules and their payments.} 
However, these auction-based systems are vulnerable to manipulations and, crucially, do not deal with the uncertainties described. Hence, in this paper, we study auction-based ridesharing mechanisms that aim to incentivize commuters in such dynamic and uncertain domains and seek to find mechanisms that are robust to manipulations. 

Similar \emph{execution uncertainty} has been addressed in task allocation domains~\cite{Porter_2008,Ramchurn_2009,Stein_2011,Feige_2011,Conitzer_2014}. However, the uncertainty modelled there is agents' ability to complete a task (i.e., whenever an agent is allocated a task, she will always incur the cost of executing the task regardless of her ability to complete it). In contrast, the uncertainty in the ridesharing context is commuters' ``willingness" rather than their ability to undertake their trips. Hence, there is no internal cost to a rider/driver if she does not want to undertake her trip. Moreover, there is no collaboration between agents for completing a task, while in ridesharing commuters have to collaborate to finish a shared trip. {In other domains, mechanisms with verification, e.g., \cite{Nisan_1999,Rose_2012,Caragiannis_2012,Fotakis_2013}, have been designed to verify agents' types after the execution of their actions. However, they are not applicable in ridesharing, because a commuter's uncertainty of undertaking her trip is temporal and is not verifiable from whether she commits.}

Against this background, we investigate incentive mechanisms for the ridesharing domain and attempt to identify scenarios in which these mechanisms will be robust to manipulations. We characterise such scenarios specifically in terms of the commuters' valuation functions (i.e., the value they attribute to rides%{\color{red}based on their knowledge of the others' probability of undertaking trips}
). By so doing, we develop a framework to study all valuation settings which, in turn, can inform the design of ridesharing booking systems. Hence, our work advances the state of the art in the following ways:
\begin{itemize}
\item We show that the Groves mechanisms are only truthful in the very special cases: either none of the commuters' valuation depends on the others' probability of undertaking their trips, or the probabilities are publicly known. 
\item Since in general settings, it is impossible to design truthful and efficient mechanisms, we propose an ex-post truthful and efficient mechanism where a commuter is rewarded if she undertakes her trip, otherwise she pays the loss she causes to the others for not undertaking her trip. Ex-post truthfulness is the best incentive we can provide here without sacrificing social welfare.
\item We then identify a sufficient and necessary condition where the proposed mechanism is ex-post truthful. This condition covers a very rich class of valuation settings in practice, but it does eliminate some interesting cases where commuters deliberately choose to not collaborate/commit under certain situations.
\end{itemize}

%Given the above criteria of the system design, we first investigate the well-known Vickrey-Clarke-Groves (VCG) mechanism and show that VCG can not be even Nash truthful (i.e., Nash equilibrium) for almost all valuation domains that are monotonically changing with respect to commuters' probabilities of commitment. However, we furthermore show that, if commuters' uncertainty of commitment is publicly known, VCG is truthful (i.e., reporting trip truthfully is a dominant strategy) for all valuation domains without externalities. A publicly known uncertainty for a commuter might represent a kind trust for the commuter, which can be calculate from the commuter's history behaviour and the ratings of their sharing experiences. Given the limited applicability of VCG, we further proposed another mechanism which pays commuters based on their realized commitment of their trips, rather than their reported probabilities of commitment which might be different from their true/realized probability of commitment. We prove that the new mechanism is Nash truthful for very broad valuation domains, including a large portion of the valuation domains that are not Nash truthful for VCG. Lastly, we present some intuitive case studies to further clarify and justify our findings.

The remainder of the paper is organized as follows. Section~\ref{sect_model} presents the ridesharing model and the desirable properties of a ridesharing system. Section~\ref{sect_vcg} investigates the applicability of the Groves mechanisms. Sections~\ref{sect_com-pay} and \ref{sect_linearity_cont} propose an new mechanism and identify a sufficient and necessary condition to truthfully implement it. We conclude in Section~\ref{sect_con}.

\section{The Ridesharing Model}
\label{sect_model}
\noindent We study a ridesharing system where there is a set of commuters each of whom has a trip they want to make and a probability that they will eventually make it. Each commuter is either a driver or a rider: a driver can either offer extra seats to riders or ride with others, while a rider can only ride with others.
% Each commuter has a cost to finish her trip, which depends on the transport she chooses, e.g. driving cost or fare of buses/taxis/trains. 
The trip (aka \textit{type}) of each commuter $i$ is modelled by $\theta_i = (v_i, p_i)$, where $v_i$ is $i$'s valuation function for receiving/offering rides and $p_i \in [0,1]$ is the probability that $i$ will undertake the trip (we call $p_i$ $i$'s \textbf{probability of commitment}). Note that, a trip normally consists of departure/arrival locations, travel times, and travel costs, which together with other travel preferences are all specified by $v_i$. Moreover, $v_i$ also specifies any opt out options such as public transports for both riders and drivers. The precise format of $v_i$ depends on the context of the application and the information available to the users. In order to cover a full range of trip types and the ability of the ridesharing systems to cope with them, we do not restrict the form of the valuation function, e.g., it may have externalities. %In the very basic scenario, $v_i$ could be modelled as simply as a fixed travel cost to finish $i$'s current trip (without sharing) within a desirable travel time window. 
Let $N$ be the set of all commuters, $\theta$ be the trip profile of $N$, $\theta_{-i}$ be the trip profile of $N$ except $i$, and $\theta = (\theta_i, \theta_{-i})$. Furthermore, let $p = (p_i)_{i \in N}$ be the profile of the probability of commitment of $N$, $p_{-i}$ be the probability profile of $N$ except $i$, and $p = (p_i, p_{-i})$.

We study ridesharing mechanisms that require each commuter to report her intended trip to the mechanism. We assume that each commuter's trip is privately observed, i.e., they do not necessarily report their trips truthfully to the mechanism if it is in their interest to do so. We denote $\theta_i$ the true trip of $i$ and $\hat{\theta}_i = (\hat{v}_i, \hat{p}_i)$ her reported trip to the mechanism. A ridesharing mechanism consists of an \textbf{allocation} policy $\pi$ and a \textbf{payment} policy $x$. 
Given the commuters' trip report profile $\hat{\theta}$, the mechanism computes an allocation $\pi(\hat{\theta}) = \{\pi_i(\hat{\theta})\}_{i\in N}$ (i.e., ridesharing schedules) and a payment $x(\hat{\theta}) = \{x_i(\hat{\theta})\}_{i\in N}$. $\pi_i(\hat{\theta}) =(r_i, s_i)$ where $r_i \in \{drive, ride, none\}$ is $i$'s \textbf{role} in the allocation and $s_i$ is the corresponding \textbf{schedule} for $i$ which specifies the times, locations, and commuters with whom $i$ will travel together on her trip. $r_i = drive$ indicates that $i$ will drive and offer rides to some riders, $r_i = ride$ indicates that $i$ will ride with other drivers, and $r_i = none$ indicates that $i$ is not scheduled to travel with the others. Note that a commuter who originally drives can be allocated to ride if that satisfies the goal of the ridesharing mechanism. $x_i(\hat{\theta}) \in \mathbb{R}$ is the payment for $i$. If $x_i(\hat{\theta}) \geq 0$, $i$ pays $x_i(\hat{\theta})$ to the system, otherwise $i$ receives $|x_i(\hat{\theta})|$ from the system. 

 % That is, the input of each commuter's $v_i$ is $(\pi(\hat{\theta}),\hat{p})$. 
Given a trip report profile $\hat{\theta}$, $\pi(\hat{\theta})$ needs to be feasible with respect to the commuters' valuations/preferences and the consistency between their schedules. We say $\pi$ is \textit{feasible} if $\pi(\hat{\theta})$ is feasible for all report profiles $\hat{\theta}$. In the rest of this paper, only feasible allocations are considered.

%\subsection{Properties of Ridesharing Mechanisms}
Other than the basic feasibility, the main goal of the ridesharing mechanisms is to maximize social welfare (in most cases this is equivalent to travel cost minimization). Since the commuters' trips are privately known, the mechanism is only able to maximize social welfare if it can get the commuters' true trips. Therefore, the mechanism needs to incentivize commuters to report their trips truthfully. Moreover, commuters should not lose when they participate in the system, i.e., they are not forced to join the system. In the following, we formally define these properties.

Given a ridesharing allocation, the expected social welfare is the sum of all commuters' valuations on the allocation. We say an allocation $\pi$ is efficient if it maximizes the expected social welfare for any trip report profile $\hat{\theta}$. To simplify notations, we separate the commuters' probability of commitment $\hat{p}$ from the allocation $\pi(\hat{\theta})$, and $i$'s valuation of the allocation $\pi(\hat{\theta})$ is given by $\hat{v}_i(\pi(\hat{\theta}), \hat{p})$.
\begin{definition}
Allocation $\pi$ is \textbf{efficient} if and only if for all trip report profiles $\hat{\theta}$, we have:
$$\sum_{i\in N} \hat{v}_i(\pi(\hat{\theta}), \hat{p}) \geq \sum_{i\in N} \hat{v}_i(\pi^{\prime}(\hat{\theta}), \hat{p})$$ where $\pi^{\prime}$ is any other feasible allocation. 
\end{definition}

Note the expected social welfare calculated by the mechanism is based on the commuters' reported trips $\hat{\theta}$ only. %, i.e., it does not depend on their true trips. %For example, $i$ might report a probability of commitment $\hat{p}_i = 1$ while her true probability is $p_i < 1$, then her expected valuation computed by the mechanism is based on $\hat{p}_i$ not $p_i$. 
However, the commuters' actual/realized valuation depends on their true trip information. That is, $i$'s realized valuation for allocation $\pi(\hat{\theta})$ is $v_i(\pi(\hat{\theta}), p)$ rather than $\hat{v}_i(\pi(\hat{\theta}), \hat{p})$, which depends on $i$'s true valuation $v_i$ and the commuters' true probability of commitment $p$.

Given the commuters' true trip profile $\theta$, reported trip profile $\hat{\theta}$ and mechanism $(\pi, x)$, commuter $i$'s expected \textbf{utility} is quasilinear and defined as: 
$$u_i(\theta_i, \hat{\theta}, \pi, x, p) = v_i(\pi(\hat{\theta}), p) - x_i(\hat{\theta}).$$
%Notice that $i$'s expected utility is based on the allocation $\pi(\hat{\theta})$, her true valuation $v_i$ and the true probabilities of commitment $p$, although the allocation $\pi(\hat{\theta})$ is always based on the reported trips $\hat{\theta}$. This is because commuters' realized valuation depends on her true trip/valuation and others' true probabilities of commitment. Moreover, $x_i(\hat{\theta})$ can also be defined with respect to the true/realized commitments of the reported trips, e.g., a driver only gets paid if she committed her trip.
We say mechanism $(\pi, x)$ is \textbf{individually rational} if $u_i(\theta_i, (\theta_i, \hat{\theta}_{-i}), \pi, x, p) \geq 0$ for all $i$, all $\theta$, and all $\hat{\theta}_{-i}$. That is, a commuter never receives a negative expected utility % (i.e., never loses money by participation) 
if she reports truthfully, no matter what others report. Furthermore, we say the mechanism is \textit{truthful} (aka \textit{dominant-strategy incentive-compatible}) if it always maximizes a commuter's expected utility if she reports her trip truthfully, i.e., reporting truthfully is a dominant strategy for all commuters.
\begin{definition}
Mechanism $(\pi, x)$ is \textbf{truthful} if and only if for all $i\in N$, all $\theta$, and all $\hat{\theta}$, we have $u_i(\theta_i, (\theta_i, \hat{\theta}_{-i}), \pi, x, p) \geq u_i(\theta_i, (\hat{\theta}_i, \hat{\theta}_{-i}), \pi, x, p).$
\end{definition}
Another solution concept weaker than dominant-strategy incentive-compatible (but still very valid) is called \textbf{ex-post truthful}, which requires that reporting truthfully maximizes a commuter's utility if everyone else also reports truthfully (i.e. reporting truthfully is an ex-post equilibrium). Ex-post truthful is stronger than \emph{Bayes-Nash truthful} which assumes that all agents know the correct probabilistic model of the distribution on their types.
%\begin{definition}
%Mechanism $(\pi, x)$ is \textbf{ex-post truthful} if for all commuter $i$, all $\theta_i$, all $\hat{\theta}_i$, all $\theta_{-i}$, we have $$u_i(\theta_i, (\theta_i, \theta_{-i}), \pi, x, p) \geq u_i(\theta_i, (\hat{\theta}_i, \theta_{-i}), \pi, x, p).$$
%\end{definition}

Given these properties defined in the above, we will study ridesharing mechanisms that are efficient, (ex-post) truthful, and individually rational.
\section{The Groves Mechanisms}
\label{sect_vcg}
\noindent In this section, we analyse the applicability of the well-known set of mechanisms called Groves mechanisms~\cite{Groves_1973} in the ridesharing domain. In many domains, the Groves mechanisms are the only mechanisms that are efficient and truthful. Groves mechanisms apply an efficient allocation $\pi^{eff}$ and charge each commuter $i$ the following:
\begin{equation}
 x_i^{Groves}(\hat{\theta}) = h_i(\hat{\theta}_{-i}) - V_{-i}(\hat{\theta}, \pi^{eff})
\end{equation} 
where
\begin{itemize}
\item $h_i$ is a function that only depends on $\hat{\theta}_{-i}$,
\item $V_{-i}(\hat{\theta}, \pi^{eff}) = \sum_{j \neq i} \hat{v}_j(\pi^{eff}(\hat{\theta}), \hat{p})$ is the social welfare for all commuters, excluding $i$, under the efficient allocation $\pi^{eff}(\hat{\theta})$.
\end{itemize}

Since $h_i$ is independent of $i$'s report, we can set $h_i(\hat{\theta}_{-i}) = 0$, and then each commuter will receive an amount equal to the social welfare of the other commuters. Thus, each commuter's utility is the social welfare of the efficient allocation, which is maximized in domains without valuation interdependence, if the commuter reports truthfully. However, in the ridesharing domain, commuters' valuations are normally interdependent via their probability of commitment. In Theorem~\ref{thm_Groves_notIC}, we show that as soon as there exists one commuter whose valuation depends on the other commuters' probability of commitment, Groves mechanisms cannot even be truthfully implemented in an ex-post equilibrium. That is, reporting truthfully is not a dominant strategy even if everyone else reports truthfully. 

Before we prove the impossibility, let's gain some intuition from an example: consider two commuters $i,j$ travelling from one location to another at the same time, and assume that only $i$ drives and the efficient allocation is to let $j$ ride with $i$. If $j$'s valuation for riding with $i$ is in the form of $\alpha_j\times p_i$ where $\alpha_j > 0$, then $i$ can increase $j$'s valuation by reporting $\hat{p}_i > p_i$ to receive a higher utility.

We say commuter $i$'s valuation is \textit{external-commit-independent} if it is independent of the probability of commitment of the other commuters. %{\color{orange}For instance, a driver does not want to ride with other drivers and also does not care how likely her riders will commit their trips (i.e., the driver's valuation/costs for a given trip does not depend on the number of riders in her car).}
\begin{definition}
Valuation $v_i$ of commuter $i$ is \textbf{external-commit-independent} if for all trip profiles ${\theta}$, all allocations $\pi$, and all probability profiles $\bar{p} = (\bar{p}_j)_{j\in N}$ where $\bar{p}_j \in [0,1]$, we have $\bar{p}_i = p_i$ implies $v_i(\pi({\theta}), \bar{p}) = v_i(\pi({\theta}), p).$
\end{definition}

\begin{theorem}
\label{thm_Groves_notIC}
The Groves mechanism is not ex-post truthful if %for all $i\in N$, $v_i$ satisfies no-externality and 
there exists $j\in N$ s.t. $v_j$ is not external-commit-independent.
\end{theorem}
\begin{proof}
Given that $j$'s valuation is not external-commit-independent, there exist a report profile $\theta$, an allocation $\pi$, and a probability profile $\bar{p} = (\bar{p}_i)_{i\in N}$ where $\bar{p}_i \in [0,1]$ such that $\bar{p}_j = p_j$ and $v_j(\pi(\theta), \bar{p}) \neq v_j(\pi(\theta), p)$. Without loss of generality assume that $\bar{p}$ only differs from $p$ in $k$'s probability of commitment, i.e., $\bar{p}_k \neq p_k$ and $\bar{p}_{-k} = p_{-k}$. 

Under efficient allocation $\pi^{eff}$, it is not hard to find a trip profile $\hat{\theta}_{-j}$ such that $\hat{p}_{-j} = p_{-j}$ and $\pi^{eff}(\theta_j, \hat{\theta}_{-j}) = \pi(\theta)$. We can choose $\hat{\theta}_{-j}$ by setting $\hat{v}_i(\pi(\theta), p)$ to a sufficiently large value for each $i\neq j$. Moreover, we require that the allocation $\pi^{eff}(\theta_j, \hat{\theta}_{-j})$ does not change if $k$ reported a different probability of commitment $\bar{p}_k$ rather than $p_k$, which can be achieved by setting $k$'s valuation $\hat{v}_k(\pi(\theta), \bar{p})$ to a sufficiently large value (no matter whether $\hat{v}_k$ is external-commit-independent).
%, i.e. the allocation for $j$ and the probability that commuters, involved in $j$'s allocation, come to share with $j$ is the same under $\theta$ and $(\theta_j, \theta_k, \hat{\theta}_{-j,k})$. Given that all valuations satisfy no-externality, the setting $\hat{\theta}_{-j,k}$ can be achieved by assigning very high valuations to the allocations in $\pi(\theta)$ related to $j$ and very low valuations to other allocations, although the trips for $j$ and $k$ are fixed.

In what follows, we show that there exist situations where commuter $k$ is incentivized to misreport. Under trip profile $(\theta_j, \hat{\theta}_{-j})$, we know that $k$ can change $j$'s valuation without changing the allocation by reporting a different probability of commitment. Regardless of the changes of the other commuters' valuations when $k$ changes her probability of commitment, there always exists a situation s.t. $v_j(\pi(\theta), p) + \sum_{i\in N\setminus\{j,k\}} \hat{v}_i(\pi(\theta), p) \neq v_j(\pi(\theta), \bar{p}) + \sum_{i\in N\setminus\{j,k\}} \hat{v}_i(\pi(\theta), \bar{p})$, even if the valuations of all commuters except $j$ are external-commit-independent, i.e., $\sum_{i\in N\setminus\{j,k\}} \hat{v}_i(\pi(\theta), p) = \sum_{i\in N\setminus\{j,k\}} \hat{v}_i(\pi(\theta), \bar{p})$. \\If $v_j(\pi(\theta), p) + \sum_{i\in N\setminus\{j,k\}} \hat{v}_i(\pi(\theta), p) < v_j(\pi(\theta), \bar{p}) + \sum_{i\in N\setminus\{j,k\}} \hat{v}_i(\pi(\theta), \bar{p})$, then commuter $k$ of true probability of commitment $p_k$ would report $\bar{p}_k \neq p_k$ to gain a better utility. %$v_j(\pi(\theta), \bar{p}) + \sum_{i\in N\setminus\{j,k\}} \hat{v}_i(\pi(\theta), \bar{p}) + \hat{v}_k(\pi(\theta), p) - h_k(\theta_j, \hat{\theta}_{-j,k})$. 
Otherwise, %if $v_j(\pi(\theta), p) + \sum_{i\in N\setminus\{j,k\}} \hat{v}_i(\pi(\theta), p) > v_j(\pi(\theta), \bar{p}) + \sum_{i\in N\setminus\{j,k\}} \hat{v}_i(\pi(\theta), \bar{p})$, then 
commuter $k$ of true probability of commitment $\bar{p}_k$ would report $p_k$ to gain a better utility. In both situations, we assume that the other commuters truthfully report their trips.
\qed
\end{proof}

Theorem~\ref{thm_Groves_notIC} shows that the Groves mechanisms \emph{cannot} be truthfully implemented in an ex-post equilibrium even if there is only one commuter whose valuation depends on the others' probability of commitment. This is a rather negative result as it says that Groves mechanisms are not applicable in all valuation settings that are interdependent via their probability of commitment. However, Theorem~\ref{thm_Groves_IC} shows that if their uncertainty of commitment is known by the mechanism, i.e., the interdependence of their valuations is known by the mechanism, then reporting valuation truthfully is still a dominant strategy in the Groves mechanisms. In some real-world applications, the commuters' probability of commitment might be computable by the ridesharing system from, say, their history participations/trips.

\begin{theorem}
\label{thm_Groves_IC}
The Groves mechanism is truthful if for all $i\in N$, %$v_i$ satisfies no-externality and 
%the probability of commitment 
$p_i$ is known by the mechanism.
\end{theorem}
%\begin{proof}
%See Appendix.\qed
%\end{proof}
\begin{proof}
According to Proposition $9.27$ from \cite{nisan_algorithmic_2007}, we need to show that for all profiles $\theta$, for all $i\in N$:
\begin{enumerate}
\item $x_i^{Groves}(\theta)$ does not depend on $\theta_i$, but only on the alternative allocation $\pi^{eff}(\theta)$. That is, for all $\hat{\theta}_i \neq \theta_i$, if $\pi^{eff}(\hat{\theta}_i, \theta_{-i}) = \pi^{eff}(\theta)$, then $x^{Groves}_i(\hat{\theta}_i, \theta_{-i}) = x^{Groves}_i(\theta)$;
\item $i$'s utility is maximized by reporting $\theta_i$ truthfully.
\end{enumerate}

%Since all commuters' valuations satisfy no-externality, given allocation $\pi(\theta)$, $i$ may only influence others' valuation via its probability of commitment $p_i$. 
Given that $p_i$ is known by the mechanism (i.e., $i$ does not need to report $p_i$), $i$ can only change others' valuations by changing the allocation, and therefore $x_i^{Groves}(\theta)$ does not depend on $\theta_i$, but only on the allocation $\pi^{eff}(\theta)$. This is not the case when $p_i$ is privately known because, as shown in Theorem~\ref{thm_Groves_notIC}, $i$ may change the other commuters' valuation without changing the allocation.

For each commuter $i$, her expected utility is $v_i(\pi^{eff}(\theta), p) - x^{Groves}_i(\theta) = v_i(\pi^{eff}(\theta), p) + V_{-i}(\theta, \pi^{eff}) - h_i(\theta_{-i})$, where the first two terms together are the social welfare and $h_i(\theta_{-i})$ is independent of $\theta_i$. Since the allocation $\pi^{eff}$ is efficient, so the social welfare and therefore $i$'s utility is maximized when $i$ reports truthfully. %Thus, $i$'s utility is maximized by reporting $\theta_i$ truthfully regardless of what others report.
\qed
\end{proof}

%As $h_i$ in the payment of the Groves mechanism is independent on $i$'s report, we can choose any value that is not related to $i$'s report, but it might give commuters negative utility if $h_i$ is too large or the system a huge deficit if $h_i$ is too small. 
It is worth mentioning that Theorems~\ref{thm_Groves_notIC} and \ref{thm_Groves_IC} do not rely on the form of $h_i$ in $x^{Groves}_i$. %For achieving other properties such as \emph{individual rationality}, 
We normally set $h_i$ to be the maximum social welfare that the others can obtain without $i$'s participation, which is known as the Clarke pivot rule (the corresponding mechanism is known as VCG). The Clarke pivot rule guarantees that all commuters' expected utilities are non-negative, i.e., it satisfies individual rationality, and also charges all commuters the maximum amount without violating individual rationality.  
\section{Commit-Based-Pay Mechanisms}
\label{sect_com-pay}
\noindent As shown in the last section, the Groves mechanisms are not applicable when the probability of commitment is privately known by the commuters. We also showed that this is due to the interdependence of the commuters' valuation created by their probability of commitment. The other reason why Groves mechanisms cannot prevent commuters' manipulations is that the Groves payment is calculated according to the commuters' reported probability of commitment rather than their realized/true probability of commitment. 

To combat this problem, one solution that has been proposed for tackling execution uncertainty in task allocation domains is that an agent is paid according the realized execution of her actions rather than what she reported~\cite{Porter_2008}. Following this principle, we define two payments for each commuter according to the realized commitment of her trip: one for successfully committing to her trip and the other for failing the commitment. The payment for successfully committing to her trip is a kind of reward, while the one when she fails works like a penalty. We call this kind of payment commit-based payment.

%In this section, we adapt the realization-based payment scheme in ridesharing and analyse its applicability.
Given the commuters' trip report profile $\hat{\theta}$ and the efficient allocation $\pi^{eff}$, the \textit{commit-based payment} $x^{com}$ for each commuter $i$ is defined as:
\begin{equation}
\label{eq_commit_pay}
 x_i^{com}(\hat{\theta}) = 
 \begin{cases}
     h_i(\hat{\theta}_{-i}) - V^1_{-i}(\hat{\theta}, \pi^{eff}) &    \text{\	     if $i$ commits her trip,}\\
     h_i(\hat{\theta}_{-i}) - V^0_{-i}(\hat{\theta}, \pi^{eff}) &      \text{\	     if $i$ does not commit her trip.}
 \end{cases}
\end{equation}
where
\begin{itemize}
\item $h_i(\hat{\theta}_{-i}) = \sum_{j \in N\setminus \{i\}} \hat{v}_j(\pi^{eff}(\hat{\theta}_{-i}), \hat{p}_{-i})$ is the maximum expected social welfare that the other commuters can achieve without $i$'s participation,
\item $V^1_{-i}(\hat{\theta}, \pi^{eff}) = \sum_{j \in N\setminus \{i\}} \hat{v}_j(\pi^{eff}(\hat{\theta}), (1, \hat{p}_{-i}))$
is the expected social welfare of all commuters except $i$ under the efficient allocation $\pi^{eff}(\hat{\theta})$ when $i$ commits. $V^0_{-i}(\hat{\theta}, \pi^{eff}) = \sum_{j \in N\setminus \{i\}} \hat{v}_j(\pi^{eff}(\hat{\theta}), (0, \hat{p}_{-i}))$ is the corresponding social welfare when $i$ fails to commit.
\end{itemize}
$x^{com}_i$ pays/rewards commuter $i$ the social welfare increased by $i$ if she commits and charges/penalizes her the social welfare loss if she does not commit. Theorem~\ref{thm_exic_com} shows that the mechanism $(\pi^{eff}, x^{com})$ is truthful in an ex-post equilibrium if all commuters' valuation is linear in commitment (Definition~\ref{def_linear}). %That is, reporting a trip truthfully is an ex-post equilibrium.

\begin{definition}
\label{def_linear}
Valuation $v_i$ of $i$ is \textbf{linear in commitment} if for all trip profiles $\theta$, all allocations $\pi$, and all $j\in N$, $v_i(\pi(\theta), p) = p_j\times v_i(\pi(\theta), (1, p_{-j})) + (1-p_j)\times v_i(\pi(\theta), (0, p_{-j}))$.
\end{definition} 
Intuitively, $v_i$ is linear in commitment if for all allocations $v_i$ is linear in the probability of commitment of all commuters including $i$ (see an example in Section~\ref{sect_linearity_cont}). 
%{\color{orange}For example, the valuation for commuter $i$ to receive a ride from a driver $j$ is $\alpha_i\times p_i\times p_j$, where $\alpha_i$ is a constant, and her valuation for offering a ride to a rider $j$ is $\beta_i\times p_i\times p_j$, where $\beta_i$ is a constant. It is evident that $i$'s valuation in this example is linear in commitment.} 
It is evident that external-commit-independent valuations are also linear in commitment.

\begin{theorem}
\label{thm_exic_com}
Mechanism $(\pi^{eff}, x^{com})$ is ex-post truthful and individually rational if for all $i\in N$, $v_i$ %satisfies no-externality and 
is linear in commitment.
\end{theorem}
\begin{proof}
Similar to the proof of Theorem~\ref{thm_Groves_IC}, we need to prove that for all $i\in N$:
\begin{enumerate}
\item $x^{com}_i$ does not depend on $i$'s report, but only on the alternative allocation;
\item $i$'s utility is maximized by reporting $\theta_i$ truthfully if the others report truthfully.
\end{enumerate}
From the definition of $x^{com}_i$ in \eqref{eq_commit_pay}, we can see that given an allocation $\pi^{eff}(\hat{\theta})$, commuter $i$ cannot change $V^1_{-i}(\hat{\theta}, \pi^{eff})$ and $V^0_{-i}(\hat{\theta}, \pi^{eff})$ without changing the allocation. Therefore, $x^{com}_i$ does not depend on $i$'s report, but only on the alternative allocation.

In what follows, we show that for each commuter $i$, if the others report trips truthfully, then $i$'s utility is maximized by reporting her trip truthfully.

Given a commuter $i$'s trip $\theta_i$ and the others' true trip profile ${\theta}_{-i}$, assume that $i$ reports $\hat{\theta}_i \neq \theta_i$. According to $x^{com}_i$, when $i$ finally commits to her trip, $i$'s utility is $u_i^1 = v_i(\pi^{eff}(\hat{\theta}_i, \theta_{-i}), (1, p_{-i})) - h_i({\theta}_{-i}) + V^1_{-i}({\theta}, \pi^{eff})$ and her utility if she fails is $u_i^0 = v_i(\pi^{eff}(\hat{\theta}_i, \theta_{-i}), (0, p_{-i})) - h_i({\theta}_{-i}) + V^0_{-i}({\theta}, \pi^{eff})$. Note that $i$'s expected valuation depends on her true valuation and the commuters' true probability of commitment $p$. Therefore, $i$'s expected utility is:
\begin{align}
p_i\times u_i^1 + (1-p_i)\times u_i^0 =& \nonumber \\
&p_i\times v_i(\pi^{eff}(\hat{\theta}_i, \theta_{-i}), (1, p_{-i})) \label{eq_1} \\
&+ (1-p_i)\times v_i(\pi^{eff}(\hat{\theta}_i, \theta_{-i}), (0, p_{-i})) \label{eq_2} \\
&+ p_i\sum_{j \in N\setminus \{i\}} {v}_j(\pi^{eff}(\hat{\theta}_i, \theta_{-i}), (1, p_{-i})) \label{eq_3} \\
&+ (1-p_i)\sum_{j \in N\setminus \{i\}} {v}_j(\pi^{eff}(\hat{\theta}_i, \theta_{-i}), (0, p_{-i})) \label{eq_4} \\
&- h_i({\theta}_{-i}). \nonumber
\end{align}
Since all valuations are linear in commitment, the sum of \eqref{eq_1} and \eqref{eq_2} is equal to $v_i(\pi^{eff}(\hat{\theta}_i, {\theta}_{-i}), p)$, and the sum of \eqref{eq_3} and \eqref{eq_4} is $\sum_{j \in N\setminus \{i\}} {v}_j(\pi^{eff}(\hat{\theta}_i, {\theta}_{-i}), p)$. Thus, the sum of \eqref{eq_1}, \eqref{eq_2}, \eqref{eq_3} and \eqref{eq_4} is the social welfare under allocation $\pi^{eff}(\hat{\theta}_i, {\theta}_{-i})$. This is maximized when $i$ reports truthfully because $\pi^{eff}$ maximizes social welfare, which is not the case when $\theta_{-i}$ is not truthfully reported.
Moreover, $h_i({\theta}_{-i})$ is independent of $i$'s report and is the maximum social welfare that the others can achieve without $i$. Therefore, by reporting $\theta_i$ truthfully, $i$'s utility is maximized and non-negative (i.e., individually rational). %Notice that, $i$'s utility does not depend on other commuters' true trips, i.e., no matter what others report, reporting truthfully is a dominant strategy for $i$.
\qed
\end{proof}

The condition of linear in commitment guarantees that $(\pi^{eff}, x^{com})$ is truthfully implemented in an ex-post equilibrium (ex-post truthful), but not in a dominant strategy (truthful). As shown in the task allocation domains considering execution uncertainty~\cite{Porter_2008,Ramchurn_2009,Stein_2011,Conitzer_2014}, ex-post truthfulness is the best we can achieve here. It is not hard to find an example where a commuter is incentivized to misreport if some commuters have misreported.
Ex-post truthfulness is also strongly applicable in domains like ridesharing because computing manipulations is both computationally hard and requiring the full knowledge of all commuters' reports. 

%Linear in commitment condition is a very natural condition on the commuters' valuation. Similar condition is implicitly assumed in task allocation domains~\cite{Porter_2008,Ramchurn_2009,Conitzer_2014}

\section{Linear in Commitment is Necessary for Truthfully Implementing $(\pi^{eff}, x^{com})$}
\label{sect_linearity_cont}
\noindent This section shows that linear in commitment condition is also necessary for $(\pi^{eff}, x^{com})$ to be ex-post truthful. We first demonstrate an intuitive example showing that if the valuations of all commuters except one are linear in commitment, then there exist settings where a commuter is incentivized to misreport in $(\pi^{eff}, x^{com})$. Then we further prove that for all commuters $i$, if $v_i$ is not linear in commitment, then there exists a setting such that $(\pi^{eff}, x^{com})$ is not ex-post truthful.

Consider a scenario of two commuters $i,j$ travelling on the same route at the same time, and assume that $i$ has a car with one extra seat to share and $j$ does not have a car to share with others. Therefore the only sharing allocation is that $j$ rides with $i$, if their total expected valuation is greater than what $i,j$ will have when they travel alone. Assume that the valuations for $i,j$ are defined as follows:
\begin{equation}
\label{eq_v_rider}
 v_i = 
 \begin{cases}
     \alpha_i\times p_i\times p_j & \text{if $i$ offers a ride to $j$,}\\
     -\infty & \text{if $i$ rides with $j$,}\\
     0 & \text{if $i$ travels alone.}
 \end{cases}
\end{equation}
where $\alpha_i \leq 0$ is a constant and represents the costs to $i$ for offering a ride to $j$.
\begin{equation}
\label{eq_v_rider2}
 v_j =
 \begin{cases}
     \beta_j\times p_i\times p_j & \text{if $j$ rides with $i$ and $p_i \geq r_j$,}\\
     0 & \text{if $j$ rides with $i$ and $p_i < r_j$,}\\
     -\infty & \text{if $j$ offers a ride to $i$,}\\
     0 & \text{if $j$ travels alone.}
 \end{cases}
\end{equation}
where $\beta_j \geq 0$ is a constant and represents the benefits, e.g., costs saved, that $j$ will receive via riding with $i$, and $r_j \in (0,1]$ is $j$'s minimum requirement on her driver's probability of commitment. %$r_j$ is a kind of reliability requirement of $j$ on her driver. 
If $p_i < r_j$, $j$ will not ride with $i$, i.e., $j$ does not want to ride with someone who is not very reliable. %does not have much guarantee to finish her trip.

It is easy to check that $v_i$ is linear in commitment, but $v_j$ is not. Assume that $p_i < r_j$, i.e. $i,j$ are not matched to share if they both report truthfully and therefore their utilities are zero. We will show that $i$ can misreport a probability of commitment $\hat{p}_i \geq r_j$ to gain a positive utility under $(\pi^{eff}, x^{com})$ if $\alpha_i\times p_i\times p_j + \beta_j\times p_i\times p_j > 0$.

Since $i$'s true probability of commitment cannot be verified by $j$ or the system from whether $i$ commits, which is especially true if their probability of commitment changes every time they travel. %, especially if they are only matched very occasionally. Moreover, even the mechanism might not be able to verify the true probability of commitment of commuters from their passed participations, especially if their probability of commitment changes every time they travel. 
Thus, in the above example, $i$ can misreport $\hat{p}_i \geq r_j$ to get matched with $j$, and $i$'s payment will be:
\begin{equation*}
 x^{com}_i = 
 \begin{cases}
      -\beta_j\times p_j  & \text{ if $i$ committed,}\\
      0 & \text{ if $i$ did not commit.}
 \end{cases}
\end{equation*}
Then $i$'s expected utility is $p_i\times ( \alpha_i\times p_j + \beta_j\times p_j) + (1-p_i)\times 0 = \alpha_i\times p_i\times p_j + \beta_j\times p_i\times p_j$. If $\alpha_i\times p_i\times p_j + \beta_j\times p_i\times p_j > 0$, $i$ is incentivized to misreport $\hat{p}_i \geq r_j > p_i$.% {\color{red}That is $(\pi^{eff}, x^{com})$ is not ex-post truthful in this scenario.}

The above example shows that even if only one commuter's valuation is not linear in commitment, there exist settings where $(\pi^{eff}, x^{com})$ is not ex-post truthful. Theorem~\ref{thm_necessary_lic} further proves that linear in commitment becomes necessary for $(\pi^{eff}, x^{com})$ to be ex-post truthful in general.
\begin{theorem}
\label{thm_necessary_lic}
If $(\pi^{eff}, x^{com})$ is ex-post truthful for all trip profiles $\theta$, then for all $i\in N$, $v_i$ is linear in commitment.
\end{theorem}
\begin{proof}
Assume that $v_i$ is not linear in commitment, i.e., there exist $\hat{\theta}_{-i}$, an allocation $\pi$, and some $j\in N$ (without loss of generality, assume that $j \neq i$) such that $v_i(\pi(\theta_i, \hat{\theta}_{-i}), (p_i, \hat{p}_{-i})) \neq \hat{p}_j\times v_i(\pi(\theta_i, \hat{\theta}_{-i}), (1, (p_i, \hat{p}_{-i})_{-j})) + (1-\hat{p}_j)\times v_i(\pi(\theta_i, \hat{\theta}_{-i}), (0, (p_i, \hat{p}_{-i})_{-j}))$. Similar to the proof of Theorem~\ref{thm_Groves_notIC}, we can find a profile ${\theta}_{-i}$ such that $p_{-i} = \hat{p}_{-i}$ and $\pi^{eff}(\theta) = \pi(\theta_i, \hat{\theta}_{-i})$. Applying $(\pi^{eff}, x^{com})$ on $\theta$, when $j$ finally commits to her trip,
$j$'s utility is $u_j^1 = v_j(\pi^{eff}(\theta), (1, p_{-j})) - h_j({\theta}_{-j}) + V^1_{-j}(\theta, \pi^{eff})$ and her utility if she fails is $u_j^0 = v_j(\pi^{eff}(\theta), (0, p_{-j})) - h_j({\theta}_{-j}) + V^0_{-j}(\theta, \pi^{eff})$. Thus, $j$'s expected utility is:
\begin{align}
p_j\times u_j^1 + (1-p_j)\times u_j^0 =& \nonumber \\
&p_j\times v_i(\pi^{eff}(\theta), (1, p_{-j})) \label{eq_1-1} \\
&+ (1-p_j)\times v_i(\pi^{eff}(\theta), (0, p_{-j})) \label{eq_2-1} \\
&+ p_j\sum_{k \in N\setminus \{i\}} v_k(\pi^{eff}(\theta), (1, p_{-j})) \label{eq_3-1} \\
&+ (1-p_j)\sum_{k \in N\setminus \{i\}} v_k(\pi^{eff}(\theta), (0, p_{-j})) \label{eq_4-1} \\
&- h_j({\theta}_{-j}). \nonumber
\end{align}
Given the non-linear in commitment assumption, \eqref{eq_1-1} and \eqref{eq_2-1} together can be written as $v_i(\pi^{eff}(\theta), p) + \delta_i$ where $\delta_i = \eqref{eq_1-1} + \eqref{eq_2-1} - v_i(\pi^{eff}(\theta), p)$. Similarly substitution can be carried out for all other commuters $k\in N\setminus \{i\}$ in \eqref{eq_3-1} and \eqref{eq_4-1} regardless of whether $v_k$ is linear in commitment. After this substitution, $j$'s utility can be written as:
\begin{align}
p_j\times u_j^1 + (1-p_j)\times u_j^0 =
\underbrace{\sum_{k \in N} v_k(\pi^{eff}(\theta), p)}_{(13)}
+ \underbrace{\sum_{k \in N} \delta_k}_{(14)}
- h_j(\theta_{-j}). \nonumber
%p_j\times u_j^1 + (1-p_j)\times u_j^0 =& \nonumber \\
%&\sum_{k \in N} v_k(\pi^{eff}(\theta), p)  \label{eq_3-2} \\
%&+ \sum_{k \in N} \delta_k \label{eq_4-2} \\
%&- h_j(\theta_{-j}). \nonumber
\end{align}
Consider a suboptimal allocation $\hat{\pi}(\theta) \neq \pi^{eff}(\theta)$, if $\hat{\pi}(\theta)$ is chosen by the mechanism, then $j$'s utility can be written as:
\begin{align}
\hat{u}_j =
\underbrace{\sum_{k \in N} v_k(\hat{\pi}(\theta), p)}_{(15)}
+ \underbrace{\sum_{k \in N} \hat{\delta}_k}_{(16)}
- h_j(\theta_{-j}). \nonumber
%\hat{u}_j =& \nonumber\\
%&\sum_{k \in N} v_k(\hat{\pi}(\theta), p)  \label{eq_3-3} \\
%&+ \sum_{k \in N} \hat{\delta}_k \label{eq_4-3} \\
%&- h_j(\theta_{-j}). \nonumber
\end{align}
In the above two utility representations, we know that terms $(13) > (15)$ %$\eqref{eq_3-2} > \eqref{eq_3-3}$ 
because $\pi^{eff}$ is efficient, but terms $(14)$ and $(16)$ %\eqref{eq_4-2} and \eqref{eq_4-3} 
can be any real numbers. In what follows, we tune the valuation of $j$ such that the optimal allocation is switching between $\hat{\pi}(\theta)$ and $\pi^{eff}(\theta)$, and $j$ is incentivized to misreport.

In the extreme case where all commuters except $i$'s valuations are linear in commitment, we have $\sum_{k \in N} \delta_k = \delta_i \neq 0$ and $\sum_{k \in N} \hat{\delta}_k = \hat{\delta}_i$ (possibly $0$). If $\delta_i > \hat{\delta}_i$, we have 
$(13) +\delta_i > (15) +\hat{\delta}_i$. %$\eqref{eq_3-2} +\delta_i > \eqref{eq_3-3} +\hat{\delta}_i$. 
In this case, we can increase $j$'s valuation for the suboptimal allocation $\hat{\pi}(\theta)$ such that $\hat{\pi}(\theta)$ becomes optimal, i.e., $(13) < (15)$, but $(13) +\delta_i > (15) +\hat{\delta}_i$ still holds. %$\eqref{eq_3-2} < \eqref{eq_3-3}$, but $\eqref{eq_3-2} +\delta_i > \eqref{eq_3-3} +\hat{\delta}_i$ still holds. 
Therefore, if $j$'s true valuation is the one that chooses $\hat{\pi}(\theta)$ as the optimal allocation, then $j$ would misreport to get allocation $\pi^{eff}(\theta)$ which gives her higher utility. If $\delta_i < \hat{\delta}_i$, then we can easily modify $j$'s valuation for $\hat{\pi}(\theta)$ such that $(13) +\delta_i < (15) +\hat{\delta}_i$ but $(13) > (15)$ still holds.  %$\eqref{eq_3-2} +\delta_i < \eqref{eq_3-3} +\hat{\delta}_i$ but $\eqref{eq_3-2} > \eqref{eq_3-3}$ still holds. 
In this case, if $j$'s true valuation again is the one just modified, then $j$ would misreport to get $\hat{\pi}(\theta)$ with a better utility. 
\qed
\end{proof}

Note that Theorem~\ref{thm_necessary_lic} does not say that given a specific profile $\theta$, all $v_i$ have to be linear in commitment for $(\pi^{eff}, x^{com})$ to be ex-post truthful. Take the example discussed in the beginning of this section, if $p_i \geq r_j$, then $i$ is not incentivized to misreport and $(\pi^{eff}, x^{com})$ is ex-post truthful, although $j$'s valuation is not linear in commitment. However, since each commuter $i$ does not know the others' trips, to truthfully implement $(\pi^{eff}, x^{com})$ in an ex-post equilibrium for all possible trips of the others, Theorem~\ref{thm_necessary_lic} says that $v_i$ has to be linear in commitment.
\section{Conclusions}
\label{sect_con}
We have explored the issue of incentive mechanism design in a ridesharing setting where commuters have uncertainty of completing their trips. We have shown that the class of Groves mechanisms are hardly applicable in this setting and therefore proposed the commit-based-pay mechanism which pays commuters according to the realization of the commitments of their trips. We have further demonstrated that the commit-based-pay mechanism is ex-post truthful, the best incentive we can provide in this setting without sacrificing social welfare, if and only if the commuters' valuations satisfy the linear in commitment condition. %We conjecture that our results can be generalized to other domains dealing with similar execution uncertainties such as task allocation domains~\cite{Porter_2008}.

Our work also leaves several directions for future research. The linear in commitment condition suggests that we need other solutions to offer incentives in settings where a commuter may only share with those commuters who have less uncertainty about their trips. Except the incentive problem, there are other important properties of the system that have not been touched in this work, especially the proposed mechanisms might run a large deficit~\cite{Myerson_1983} and the scheduling problem is computationally hard. Moreover, in real-world applications, commuters might not have the perfect knowledge of their travel uncertainty and we may consider discretizing the uncertainty. 

\bibliographystyle{splncs}
\bibliography{ijcai15rs}

\end{document}